\documentclass[aps,11pt,twoside, nofootinbib, superscriptaddress]{revtex4}
\usepackage{amsthm}
\usepackage{amsmath,latexsym,amssymb,verbatim,enumerate,graphicx}
\usepackage{color}

\newtheorem{definition}{Definition} %[section]
\newtheorem{prop}[definition]{Proposition}
\newtheorem{lemma}[definition]{Lemma}

\newtheorem{thm}[definition]{Theorem}

\newtheorem*{rep@theorem}{\rep@title}
\newcommand{\newreptheorem}[2]{%
\newenvironment{rep#1}[1]{%
 \def\rep@title{#2 \ref{##1} (restatement)}%
 \begin{rep@theorem}}%
 {\end{rep@theorem}}}
\makeatother

\newreptheorem{thm}{Theorem}
\newreptheorem{lem}{Lemma}

\def\ba#1\ea{\begin{align}#1\end{align}}
\def\ban#1\ean{\begin{align*}#1\end{align*}}

\newcommand{\be}{\begin{equation}}
\newcommand{\ee}{\end{equation}}

\def\benum{\begin{enumerate}}
\def\eenum{\end{enumerate}}

\def\squareforqed{\hbox{\rlap{$\sqcap$}$\sqcup$}}
\def\qed{\ifmmode\squareforqed\else{\unskip\nobreak\hfil
\penalty50\hskip1em\null\nobreak\hfil\squareforqed
\parfillskip=0pt\finalhyphendemerits=0\endgraf}\fi}
\def\endenv{\ifmmode\;\else{\unskip\nobreak\hfil
\penalty50\hskip1em\null\nobreak\hfil\;
\parfillskip=0pt\finalhyphendemerits=0\endgraf}\fi}
\newcommand{\<}{\langle}
\renewcommand{\>}{\rangle}

\def\be{\begin{equation}}
\def\ee{\end{equation}}
\def\ben{\begin{eqnarray}}
\def\een{\end{eqnarray}}

\def\bei{\begin{itemize}}
\def\eei{\end{itemize}}

\def\ep{\epsilon}

\mathchardef\ordinarycolon\mathcode`\:
\mathcode`\:=\string"8000
\def\vcentcolon{\mathrel{\mathop\ordinarycolon}}
\begingroup \catcode`\:=\active
  \lowercase{\endgroup
  \let :\vcentcolon
  }

\newcommand{\nc}{\newcommand}
 \nc{\proj}[1]{|#1\rangle\!\langle #1 |} 
\nc{\avg}[1]{\langle#1\rangle}

\nc{\conv}{\operatorname{conv}}
\nc{\smfrac}[2]{\mbox{$\frac{#1}{#2}$}} \nc{\Tr}{\operatorname{Tr}}
\nc{\ox}{\otimes} \nc{\dg}{\dagger} \nc{\dn}{\downarrow}
\nc{\lmax}{\lambda_{\text{max}}}
\nc{\lmin}{\lambda_{\text{min}}}

\nc{\csupp}{{\operatorname{csupp}}}
\nc{\qsupp}{{\operatorname{qsupp}}} \nc{\var}{\operatorname{var}}
\nc{\rar}{\rightarrow} \nc{\lrar}{\longrightarrow}
\nc{\poly}{\operatorname{poly}}
\nc{\polylog}{\operatorname{polylog}} \nc{\Lip}{\operatorname{Lip}}
\nc{\Om}{\Omega}
\nc{\wt}[1]{\widetilde{#1}}

\def\>{\rangle}
\def\<{\langle}

\nc{\glneq}{{\raisebox{0.6ex}{$>$}  \hspace*{-1.8ex} \raisebox{-0.6ex}{$<$}}}
\nc{\gleq}{{\raisebox{0.6ex}{$\geq$}\hspace*{-1.8ex} \raisebox{-0.6ex}{$\leq$}}}

\nc{\vholder}[1]{\rule{0pt}{#1}}
\nc{\wh}[1]{\widehat{#1}}
\nc{\h}[1]{\widehat{#1}}

\nc{\ob}[1]{#1}

\def\beq{\begin {equation}}
\def\eeq{\end {equation}}

\def\be{\begin{equation}}
\def\ee{\end{equation}}

\nc{\eq}[1]{(\ref{eq:#1})} 
\nc{\eqs}[2]{\eq{#1} and \eq{#2}}

\nc{\eqn}[1]{Eq.~(\ref{eqn:#1})}
\nc{\eqns}[2]{Eqs.~(\ref{eqn:#1}) and (\ref{eqn:#2})}

\nc{\region}{\cS\cW}
\newenvironment{protocol*}[1]
  {
    \begin{center}
      \hrulefill\\
      \textbf{#1}
  }
  {
    \vspace{-1\baselineskip}
    \hrulefill
    \end{center}
  }

\begin{document}

%\preprint{}

\title{Robust Device Independent Randomness Amplification}

%\title{Free randomness amplification using bipartite chain correlations}% Force line breaks with \\
\author{Ravishankar Ramanathan}
\affiliation{National Quantum Information Center of Gda\'{n}sk, 81-824 Sopot, Poland}
\affiliation{Institute of Theoretical Physics and Astrophysics, University of Gda\'{n}sk, 80-952 Gda\'{n}sk, Poland}

\author{Fernando G. S. L. Brand\~{a}o}
\affiliation{Department of Computer Science, University College London}

\author{Andrzej Grudka}
\affiliation{Faculty of Physics, Adam Mickiewicz University, 61-614 Pozna\'{n}, Poland}

\author{Karol Horodecki}
\affiliation{National Quantum Information Center of Gda\'{n}sk, 81-824 Sopot, Poland}
\affiliation{Institute of Informatics, University of Gda\'{n}sk, 80-952 Gda\'{n}sk, Poland}

\author{Micha{\l} Horodecki}
\affiliation{National Quantum Information Center of Gda\'{n}sk, 81-824 Sopot, Poland}
\affiliation{Institute of Theoretical Physics and Astrophysics, University of Gda\'{n}sk, 80-952 Gda\'{n}sk, Poland}

\author{Pawe{\l} Horodecki}
\affiliation{National Quantum Information Center of Gda\'{n}sk, 81-824 Sopot, Poland}
\affiliation{Faculty of Applied Physics and Mathematics, Technical University of Gda\'{n}sk, 80-233 Gda\'{n}sk, Poland}

%\author{}

%\affiliation{}

%\author{}
%\affiliation{}

\date{\today}% It is always \today, today,

\begin{abstract}
In randomness amplification a slightly random source is used to produce an improved random source. Perhaps surprisingly, a single source of randomness cannot be amplified at all classically. However, the situation is different if one considers correlations allowed by quantum mechanics as an extra resource. Here we present a protocol that amplifies Santha-Vazirani sources arbitrarily close to deterministic into fully random sources. The protocol is device independent, depending only on the observed statistics of the devices and on the validity of the no-signaling principle between different devices. It improves previously-known protocols in two respects. First the protocol is tolerant to noise so that even noisy quantum-mechanical systems give rise to good devices for the protocol. Second it is simpler, being based on the violation of a four-party Bell inequality and on the XOR as a hash function. As a technical tool we prove a new de Finetti theorem where the subsystems are selected from a Santha-Vazirani source. 
\end{abstract}

%\pacs{03.67.Lx, 42.50.Dv}% PACS, the Physics and Astronomy
\maketitle

\section{Introduction}
Inferring the presence of completely random processes in nature is of both fundamental and practical importance, with applications ranging from cryptography and numerical simulations to gambling. Even though in applications one usually needs a source of nearly perfect random bits (unbiased and uncorrelated with anything else), in practice only imperfect randomness is available. Is there a way of amplifying the quality of a source of randomness? Perhaps surprisingly the answer is negative in the classical world \cite{SV}: one can never amplify the randomness of a unique random source \footnote{If one has access to two or more \textit{independent} sources, randomness amplification is possible classically.}.

As randomness sources we consider Santha-Vazirani sources \cite{SV}, defined by the property that for any bit string $X = (X_1, X_2, \dots, X_n)$ produced by the source and for any $1 \leq i < n$, 
\be
\frac{1}{2} - \varepsilon \leq p(X_{i+1} = 0 | X_i, \dots, X_1) \leq \frac{1}{2} + \varepsilon. 
\label{SVcond}
\ee 
Thus each bit produced by a $\varepsilon$-SV source can be seen as the flip of a biased coin, with the bias determined by the history of the process, but always upper bounded by $\varepsilon$. The goal of randomness amplification is to use an $\varepsilon$-SV source to produce a bit that is as close as possible to a fully random bit. In \cite{SV} it was proved that for every $\varepsilon' < \varepsilon$ there is no protocol transforming an $\varepsilon$-SV source into an $\varepsilon'$-SV source.

The impossibility result of \cite{SV} only holds for classical protocols, leaving open the possibility that with non-classical resources randomness amplification might be possible. Indeed the violation of Bell inequalities by quantum correlations implies that the measurement outcomes could not have been predetermined, so one may be tempted to conclude that Bell experiments already achieve randomness amplification. However, this conclusion is marred by the fact that the Bell tests also require measurement settings to be chosen randomly; without this measurement independence, it is possible to construct deterministic models to explain the Bell violation \cite{Barrett, Hall}. 

In a seminal work, Colbeck and Renner \cite{Renner} used the violation of Bell inequalities by quantum correlations to infer that, in contrast to the classical case, randomness amplification is possible for $\varepsilon$-SV sources for a certain range of $\varepsilon$. The main idea was to use the imperfect random bits from the SV source to choose the measurement settings of a set of spatially separated observers in a Bell test. The only assumption made was the validity of the no-signaling principle; in fact this was shown to be also necessary for perfect randomness to occur in any theory \cite{Renner}.  

Ref. \cite{Renner} left open the question of whether any SV source, as long as not fully deterministic, could be amplified into a perfect source. This was answered in the affirmative by Gallego \textit{et al.} in \cite{Acin}, where it was shown that even from an arbitrarily small non-zero amount of randomness in the SV source, one may obtain perfectly random bits by using quantum correlations that violate a five-party Bell inequality. However the protocol does not tolerate noise and requires a large number of devices for its implementation. See also \cite{Grudka, Pawlowski, Alastair, Scarani, Plesch} for more recent work in the area. 

It may be helpful to distinguish randomness amplification from the task of (device-independent) randomness expansion, where one assumes that an input seed of perfect random bits is available and the goal is to expand a given random bit string into a larger sequence of random bits. Quantum non-locality has found application also in this latter task \cite{Colbeck10, Pironio, Colbeck, Acin2, Fehr11, Pironio13, VV12a} as well as in device-independent cryptographic scenarios (see e.g. \cite{BHK, Masanes09, Hanggi, VV12b, Coudron}).

\subsection{Result}

In this paper we present a protocol for randomness amplification secure against no-signalling adversaries that can tolerate a constant rate of noise depending only on the quality of the initial random source. 

\begin{thm}[informal] For every $\varepsilon > 0$, there is a protocol using an $\varepsilon$-SV source and $O(\log(1/\varepsilon'))$ non-signalling devices that with high probability either produces a bit that is $\varepsilon'$-close to uniform or aborts. Moreover single-qubit local measurements on many-copies of a four-qubits entangled quantum state, with $\poly(1 - 2\varepsilon)$ error rate (per qubit of the state or measurement) give rise to devices that do not abort the protocol with high probability.
\end{thm}

See Proposition \ref{correctnessprot1} for more a precise formulation of the result. 

As we discuss in the next subsection, the protocol is also simpler than the one presented in \cite{Acin}. First it is based on the violation of a four-partite Bell inequality, instead of a five-partite one. Second it uses the XOR as a hash function, in place of a random function.

\subsection{Overview of the Protocol and of its Correctness Proof} \label{overview}

Here we give an overview of the protocol for randomness amplification and give a high-level explanation of why it works. The protocol is depicted in Fig \ref{no-signaling-fig} and outlined in Fig \ref{Protocol}. It involves $4k$ non-signalling boxes. Each box can be reused several times and it is assumed that there are no signalling from the future to the past inside each box (see section \ref{assumptions}). 

We split the boxes into $k$ groups of 4 boxes each. We call each group a \textit{device}. In each device the goal is to violate a particular four-partite Bell inequality, with 2 inputs and outputs for each of the 4 parties. The inequality is stated explicitly in section \ref{bellinequality}. Its violating "maximally" when the value is zero (the minimum possible value). Here we are mostly concerned with the following interesting property that makes the inequality useful for randomness amplification: Given any non-signalling box violating the inequality maximally, the bit $(p, 1-p)$ corresponding to the majority of the first three output bits, given \textit{any} input, is such that $1/4 \leq p \leq 3/4$. If the value of the Bell inequality is only $\delta$, then we have 
\begin{equation}
1/4 - O(\delta) \leq p \leq 3/4 + O(\delta).
\end{equation}
This fact is established by a linear programming argument (similarly to \cite{Acin}) in Lemma \ref{lin-prog}. Therefore if we had the promise that the box at hand violates the Bell inequality close to maximally, one could extract (imperfect) certified randomness by choosing a deterministic input. 

Suppose further that we had $k$ devices with Bell value approximately zero and that we were promised the devices were product with each other, with the outputs of each device not depending on either the inputs or the outputs of the others. Then by choosing deterministic inputs to all the devices one would obtain $k$ independent bits with distributions $(p_i, 1 - p_i)$ satisfying $1/4 - O(\delta) \leq p_i \leq 3/4 + O(\delta)$. As we show in Lemma \ref{XOR-lemma}, computing the XOR of the bits one would obtain a bit $(q, 1 - q)$ such that
\begin{equation}
1/2 -  (1/4 + O(\delta))^k \leq q \leq 1/2 + (1/4 + O(\delta))^k,
\end{equation}
which is arbitrarily close to uniform for $k$ sufficiently large.  However typically the devices will not be independent of each other and there is no a priori guarantee that the boxes violate the Bell inequality. 

Let us first address the first challenge that the devices might not be product . An important ingredient of the protocol is a procedure to reduce the general case to the case of uncorrelated boxes. To this goal we prove a new version of the de Finetti theorem (see Lemma \ref{deFinetti-bound2}). Suppose we choose $(a_1, \ldots, a_k) \in [n_1] \times \ldots \times [n_k]$ \footnote{$[n]$ stands for the set $\{1, \ldots, n \}$.} from an $\varepsilon$-SV source (with $n_1, \ldots, n_k$ fixed integers) and reuse the $j$-th device $a_j$ times. Then the de Finetti bound says that the distribution of the $k$ devices in their last use (given by $a_j$ for the $j$-th device), conditioned on the inputs and outputs of all previous uses, is close to independent (as long as $n_1, \ldots n_k$ are chosen large enough, with their sizes increasing with $1/2-\varepsilon$). The proof of this bound is based on the information-theoretical approach of \cite{Brandao, Brandao2} and might be of independent interest. Thus after reusing many times the devices we are in a situation of having $k$ (close to) uncorrelated devices. 

The second challenge consists in ensuring that most of the $k$ devices, which are approximately uncorrelated, violate the Bell inequality close to maximally. In order to do so we perform a statistical test to try to estimate the average Bell inequality violation of them, and reject if such value is not sufficiently close to zero. Here we have to address the difficulty that the inputs to the devices are chosen from a $\varepsilon$-SV source, while the value of the Bell inequality that we would like to estimate is calculated over uniform input. This is the content of Lemma \ref{Azuma-estimation}.

In the next sections we provide a full description of the protocol and its correctness proof.

\begin{figure}
\scalebox{0.45}
{\includegraphics{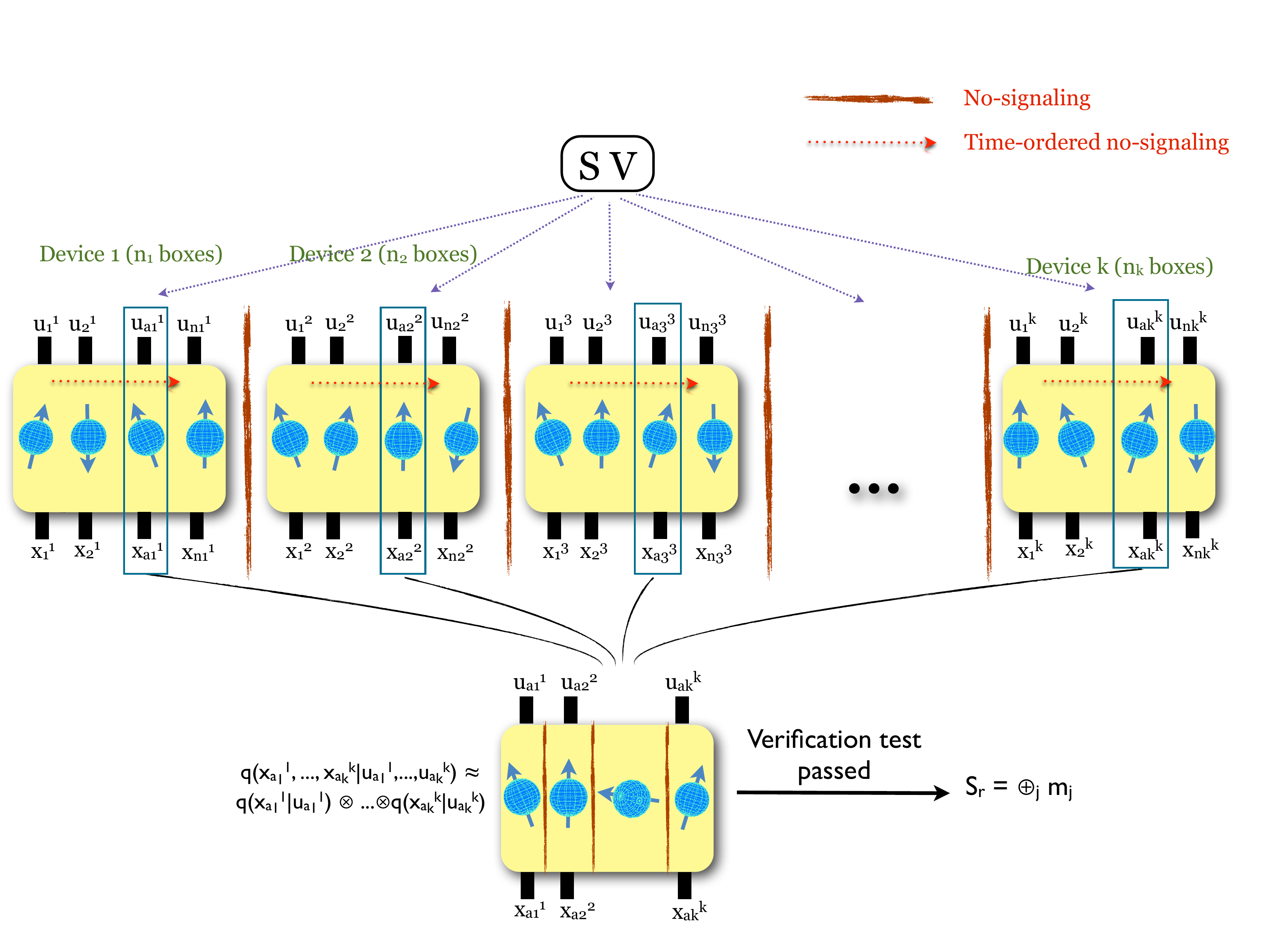}}
\caption{Illustration of the protocol for randomness amplification}
\label{no-signaling-fig}
\end{figure}

\vspace{0.3 cm}

\section{Preliminaries}
In this section we state the Bell inequality we will employ, explain what are the assumptions we impose on the devices, and then define the notion of randomness relative to an eavesdropper that we will consider.

\subsection{The Bell inequality} \label{bellinequality}
The inequality we consider for the task of randomness amplification involves four spatially separated parties with measurement settings $\textbf{u} = \{u^1, u^2, u^3, u^4\}$ and respective outcomes $\textbf{x} = \{x^1, x^2, x^3, x^4\}$. Each party chooses one of two measurement settings with two outcomes each so that $u^i \in \{0, 1\}$ and $x^i \in \{0,1\}$ for $i \in \{1,..,4\}$. Half of the $2^4$ possible measurements enter the inequality and these can be divided into two sets 
\begin{equation}
\textbf{U}_0 = \{ \{0001\}, \{0010\}, \{0100\}, \{1000\} \} \hspace{0.3 cm} \text{and} \hspace{0.3 cm} \textbf{U}_1 = \{ \{0111\}, \{1011\}, \{1101\}, \{1110\} \}. 
\end{equation}
The inequality is then \cite{Guhne}
\begin{eqnarray}
\label{Bell-ineq}
\sum_{\textbf{x}, \textbf{u}} (\texttt{I}_{\oplus_{i=1}^{4} x^i = 0} \; \texttt{I}_{\textbf{u} \in \textbf{U}_0} \; + \texttt{I}_{\oplus_{i=1}^{4} x^i = 1} \; \texttt{I}_{\textbf{u} \in \textbf{U}_1}) \; P(\textbf{x}|\textbf{u}) \geq 2,
\end{eqnarray}
where the indicator function $\texttt{I}_{L} = 1$ if $L$ is true and $0$ otherwise. The local hidden variable bound is $2$ and there exist no-signaling distributions that reach the algebraic limit of $0$. For any no-signaling box represented by a vector of probabilities $\{P(\textbf{x}|\textbf{u})\}$, the Bell inequality may be written as
\begin{equation}
\textbf{B}.\{P(\textbf{x} | \textbf{u})\} = \sum_{\textbf{x}, \textbf{u}} B(\textbf{x} , \textbf{u}) P(\textbf{x} | \textbf{u}) \geq 2, 
\end{equation}
where $\textbf{B}$ is an indicator vector for the Bell inequality with $2^4 \times 2^4$ entries 
\begin{equation}
B(\textbf{x} , \textbf{u}) = \texttt{I}_{\oplus_{i=1}^{4} x^i = 0} \; \texttt{I}_{\textbf{u} \in \textbf{U}_0} + \texttt{I}_{\oplus_{i=1}^{4} x^i = 1} \; \texttt{I}_{\textbf{u} \in \textbf{U}_1}.
\end{equation} 

Consider the quantum state 
\begin{equation} \label{state}
|\Psi \rangle = \frac{1}{\sqrt{2}} (|\phi_{-}\rangle |\tilde{\phi}_{+} \rangle + |\psi_{+} \rangle |\tilde{\psi}_{-} \rangle),
\end{equation}
where $|\phi_{-} \rangle = \frac{1}{\sqrt{2}}(|0\rangle |0\rangle - |1\rangle |1\rangle)$, $|\psi_{+}\rangle = \frac{1}{\sqrt{2}} (|0\rangle |1 \rangle + |1\rangle |0\rangle)$, $|\tilde{\phi}_{+}\rangle = \frac{1}{\sqrt{2}}(|0\rangle |+\rangle + |1\rangle |-\rangle)$, and $|\tilde{\psi}_{-}\rangle = \frac{1}{\sqrt{2}}(|0\rangle |-\rangle - |1\rangle |+\rangle)$. Measurements in the $X$ basis 
\begin{equation} \label{mea1}
\{|+\rangle = \frac{1}{\sqrt{2}} (|0\rangle + |1\rangle), |-\rangle = \frac{1}{\sqrt{2}}(|0\rangle - |1\rangle)\}
\end{equation}
correspond to $u^i = 0$ and measurements in the $Z$ basis 
\begin{equation} \label{mea2}
\{|0\rangle, |1\rangle\} 
\end{equation}
correspond to $u^i = 1$ for each of the four parties $i \in \{1, \dots, 4\}$. These measurements on $|\Psi \rangle$ lead to the algebraic violation of the inequality, i.e., the sum of the probabilities appearing in the inequality is zero. 

The reason for the choice of this Bell inequality is twofold. Firstly, as we have seen, there exist quantum correlations achieving the maximal no-signaling violation of the inequality, which implies that free randomness amplification starting from any initial $\epsilon$ of the SV source may be possible. Secondly, we will show (in Lemma \ref{lin-prog}) that for any measurement setting appearing in the inequality $\textbf{u} \in \textbf{U}_0 \cup \textbf{U}_1$, the majority function of the first three outputs $\text{maj}(x^1, x^2, x^3)$ (where $\text{maj}(x^1,x^2,x^3) = 0$ if at least two of $\{x^1, x^2, x^3 \}$ are $0$ and $1$ otherwise) cannot be predicted with certainty, given any no-signaling box that violates the inequality close to its algebraic maximum value.

\subsection{Assumptions on the Devices} \label{assumptions}

In the protocol, we consider $4k+1$ devices that cannot signal among themselves and assume that one of the devices is help by an external party (for example, an eavesdropper Eve). We describe the correlations in the $4k$ devices by the following joint probability distribution:
\begin{equation}
P(\textbf{x}^1_{\leq n_1}, \ldots, \textbf{x}^k_{\leq n_k},  \mathbb{Z} | \textbf{u}^1_{\leq n_1}, \ldots, \textbf{u}^k_{\leq n_k}, \mathbb{W}),
\end{equation}
where e.g. $\textbf{x}^1_{\leq n_1}$ denotes the vector $(\textbf{x}^1_{1}, \ldots, \textbf{x}^1_{n_1})$. For device $j$ with $1 \leq j \leq k$, the inputs and outputs of the $l_j$-th box (associated to the $l_j$-th use of the device) are given by  $\textbf{u}^j_{l_j}$ and $\textbf{x}^j_{l_j}$, respectively. Likewise the input and output of the last device (held by the eavesdropper) are given by $\mathbb{W}$ and $\mathbb{Z}$. One may think of the adversary holding the set of random variables $\mathbb{Z}, \mathbb{W}$ and supplying $4k$ devices which produce the conditional probability distribution $P$ whose behavior depends on $\mathbb{Z}$ and $\mathbb{W}$.  We say $P$ is a $(4k; n_1, \ldots, n_k)$ time-ordered non-signaling box, meaning that it is non-signaling between each of the $4k$ devices and time-ordered non-signaling within each device. 

As in previous works on randomness amplification \cite{Renner, Acin, Pawlowski, Grudka, Scarani}, we make the assumption that the $\varepsilon$-SV source and the boxes can be correlated with each other only through the no-signaling adversary Eve. In other words, we assume that the SV source, Eve's random variables $(\mathbb{Z}, \mathbb{W})$ and the box $P$ constitute a Markov chain, so that given $(\mathbb{Z}=z, \mathbb{W}=w)$, the box $P(\textbf{x}^1_{\leq n_1}, \ldots, \textbf{x}^k_{\leq n_k} | \textbf{u}^1_{\leq n_1}, \ldots, \textbf{u}^k_{\leq n_k}, \mathbb{Z} =z, \mathbb{W} = w)$ is independent of the bits produced by the source. 

%This reflects the fact that the source and box could have been prepared by Eve who may have correlated them with the system at her disposal. 

%Note that the conditional independence of the source and box distributions does not preclude the fact that these could be highly correlated but simply reflects that the correlations are bounded (by $\varepsilon$) so that the adversary cannot predict the bits produced by the source with perfect accuracy and we are indeed provided with initial randomness. 

%Moreover, the conditional independence of source and box given an adversary also does not allow to achieve randomness amplification classically. 

\subsection{Randomness Criterion}

To quantify the quality of the output we will use of the distance to uniform of a random variable $S \in \Sigma$, conditioned on Eve's input and output:
\begin{eqnarray}
\label{distance-exp}
\textit{d}( S | \mathbb{Z}, \mathbb{W}) := \frac{1}{2} \max_{s, w, z}  |P(S = s| \mathbb{Z} = z,\mathbb{W} = w) - 1/|\Sigma||.  
\end{eqnarray}

Although this function is convenient to work, it is not universally composable. A better definition of randomness relative to an eavesdropper is the following (see e.g. \cite{Hanggi}):
\begin{eqnarray}
\label{distance-exp2}
\textit{d}_{c}( S | \mathbb{Z}, \mathbb{W}) := \frac{1}{2} \sum_{s=0}^{|\Sigma|-1}  \max_{w} \sum_{z} P( z | w )  |P( S = s | \mathbb{Z} = z,\mathbb{W} = w) - 1/|\Sigma||, 
\end{eqnarray}
with $|\Sigma|$ the size of $\Sigma$.

However there is the following relation between them:
\begin{lemma} \label{conversioncomposable}
For a random variable $S \in \Sigma$,
\begin{equation}
\textit{d}_{c}( S | \mathbb{Z}, \mathbb{W}) \leq |\Sigma| \hspace{0.02 cm} \textit{d}( S | \mathbb{Z}, \mathbb{W}).
\end{equation}
\end{lemma}
\begin{proof}
\begin{eqnarray}
&& \frac{1}{2} \sum_{s=0}^{|\Sigma|-1}  \max_{w} \sum_{z} P( z | w )  |P( S = s | \mathbb{Z} = z,\mathbb{W} = w) - 1/|\Sigma|| \nonumber \\  
&\leq& \frac{|\Sigma|}{2} \hspace{0.02 cm}  \max_{s, w}    \sum_{z} P( z | w )  |P( S = s | \mathbb{Z} = z,\mathbb{W} = w) - 1/|\Sigma|| \nonumber \\  
&\leq& \frac{|\Sigma|}{2} \hspace{0.02 cm}  \max_{s, w, z}  |P( S = s | \mathbb{Z} = z,\mathbb{W} = w) - 1/|\Sigma|| \nonumber \\
&=& |\Sigma| \hspace{0.02 cm} \textit{d}( S | \mathbb{Z}, \mathbb{W}).
\end{eqnarray}
\end{proof}

%The quantity relevant for the definition of randomness is the distance from uniform of the output bit $\textsl{S}_{r}$ of the protocol (given in section \ref{Protocol}) conditioned on Eve's information. We define this distance as:
%\begin{eqnarray}
%\label{distance-exp}
%\textit{d}(\textsl{S}_{r}| \mathbb{Z}, \mathbb{W}) := \frac{1}{2} \max_{w, z} \sum_{s=0,1} |P(S_{r}=s| \mathbb{Z} = z,\mathbb{W} = w) - 1/2|.  
%\end{eqnarray}
%This takes into account the most general attack of the eavesdropper (called a coherent attack), Eve may choose any input $W=w$ and output $Z=z$ to her part of the box, choosing these so as to maximize the distance of $\textsl{S}_{r}$ from uniform. 

%\subsection{The Protocol for Randomness Amplification}
%\label{Protocol}

%The protocol for randomness amplification from an $\varepsilon$-SV source is given in Fig.~\ref{protocol} and illustrated in Fig.~\ref{no-signaling-fig}. We denote by $\textbf{u}^i_{\leq n}$ the bit string $(\textbf{u}^i_{1}, \ldots, \textbf{u}^i_{n})$.

%The protocol is based on correlations violating a four-party Bell inequality, includes an explicit hash function (the XOR of majority bits) in addition to being robust to noise in the Bell inequality violation as a function of the initial $\varepsilon$. It produces an $\varepsilon'$-secure random bit $\textsl{S}_r$ with the use of $\textit{O}(\log{(1/\varepsilon')})$ devices.

\begin{figure}[h]
\begin{protocol*}{Protocol}
\begin{enumerate}
%\item The four parties share $N$ copies of the state $|\Psi\rangle$ (or indeed any set of correlations that generates a box violating the Bell inequality algebraically).
\item The $\varepsilon$-SV source generates $4M$ bits which are used by the parties to choose the measurement settings $\textbf{u}^1_{\leq m_1}, \ldots, \textbf{u}^k_{\leq m_k}$ for the $4k$ devices, where $M = \sum_{j=1}^k m_j$. The measurements are performed sequentially on each device. The devices produce the output bits $\textbf{x}^1_{\leq m_1}, \ldots, \textbf{x}^k_{\leq m_k}$. After the measurements, the parties discard the settings that do not appear in the Bell inequality in Eq. (\ref{Bell-ineq}) and a number $n_j$ of runs remain in the $j$-th device for $1 \leq j \leq k$. 
%The $n$ remaining outcomes obey the probability distribution (box) $P(\textbf{x}_1, \ldots, \textbf{x}_n | \textbf{u}_1, \ldots, \textbf{u}_n)$.
\item The parties choose one box $a_j$ from each device using $\log{n_j}$ bits from the $\varepsilon$-SV source, for $1 \leq j \leq k$. 
%The chosen devices form a box $P(\textbf{x}^1_{a_1}, \ldots, \textbf{x}^k_{a_k} | \textbf{u}^1_{a_1}, \ldots, \textbf{u}^k_{a_k})$. 
\item The parties perform an estimation of the violation of the Bell inequality in the chosen boxes by computing the empirical average $Z_k :=  \frac{1}{k} \sum_{i=1}^{k} B(\textbf{x}_{a_i}, \textbf{u}_{a_i})$. The protocol is aborted unless $Z_k \leq \left( \frac{1}{2} - \varepsilon  \right)^4 \frac{\delta}{2} (1-\mu)$ (with fixed constants $\delta, \mu$).
\item Conditioned on not aborting in the previous step, the parties compute $\textsl{m}_j := \text{maj}(x^1_j, x^2_j, x^3_j)$ for all $1 \leq j \leq k$, and output the bit $\textsl{S}_r = \oplus_{j=1}^{k} \textsl{m}_j$. 
\end{enumerate}
\end{protocol*}
\caption{Protocol for device-independent randomness amplification}
\label{protocol}
\end{figure}

\section{Proof of correctness of the Protocol}
%\subsection{Proof of Proposition \ref{correctnessprot1}}

We will be interested in the protocol with the following parameters:
\begin{equation} \label{choicen}
n_1 = 1, \hspace{0.3 cm} n_i^{1-\log{(1+2\varepsilon)}} = 8 \ln{(2)} k^2 t^3 n_{i-1}.
\end{equation}
for a parameter $t > 0$.

\begin{prop} \label{correctnessprot1}
Let $(n_1, \ldots, n_k)$ be given by Eq. (\ref{choicen}). Then conditioned on not aborting, the bit $\textsl{S}_r$ produced by the protocol is such that 
\begin{equation}
\label{no-signaling-dist}
\textit{d}_c(\textsl{S}_{r}| \mathbb{Z}, \mathbb{W}) \leq \left( \frac{11 + 7 \delta}{16} \right)^{\mu k}  + 2 e^{-\frac{k \left( \frac{1}{2} - \varepsilon  \right)^8 (1 - \mu)^2 \delta^2}{8} } + 4/\sqrt{t}.
\end{equation}
Moreover if the box is realized by performing measurements which are $O(\delta (1-\mu) (1 - 2\varepsilon)^4)$-close to either one of the measurements of Eqs. (\ref{mea1}, \ref{mea2}) on states which are $O(\delta (1-\mu) (1 - 2\varepsilon)^4)$-close to the state of Eq. (\ref{state}), then the protocol accepts with high probability.
\end{prop}

\noindent \textit{Remark:} As a corollary of the theorem we find that using $O(l \log(l/\varepsilon))$ devices we can extract $l$ bits which are $\varepsilon$-close in variational distance to $l$ uniform bits.

\begin{proof} 

First we apply the results of Lemma \ref{deFinetti-bound2} to the time-ordered no-signaling box $P(\textbf{x}^1_{\leq n_1}, \ldots, \textbf{x}^k_{\leq n_k} | \textbf{u}^1_{\leq n_1}, \ldots, \textbf{u}^k_{\leq n_k}, \mathbb{Z}=z, \mathbb{W}=w)$. Under the Markov assumption that the SV source and this box are uncorrelated, Lemma \ref{deFinetti-bound2} gives that the $k$ boxes chosen one from each device with the $\varepsilon$-SV source are uncorrelated with high probability. For uncorrelated boxes, 
\begin{equation}
\delta_l =  \textbf{B}.\{P_{  \textbf{u}_{<l},   \textbf{x}_{<l}   }(\textbf{x}_l | \textbf{u}_l) \} = \textbf{B}.\{P(\textbf{x}_l | \textbf{u}_l) \},
\end{equation}
for $1 \leq l \leq k$. Lemma \ref{Azuma-estimation} applied to these $k$ boxes, in turn, implies that when the test accepts, $(q(\textbf{x}^1_{a_1}, \dots, \textbf{x}^k_{a_k} | \textbf{u}^1_{a_1}, \dots, \textbf{u}^k_{a_k}), (\textbf{u}^1_{a_1}, \dots, \textbf{u}^k_{a_k})), (\textbf{x}^1_{a_1}, \dots, \textbf{x}^k_{a_k})$ are $(\mu, \delta)$ good with high probability. In other words, when the test accepts, with high probability a fraction $\mu k$ of the boxes has Bell value smaller than $\delta$. We may therefore, up to the error in the de Finetti bound (given by Eq. (\ref{deFinetti})) and in the verification procedure (given in Lemma \ref{Azuma-estimation}), apply Lemma \ref{lin-prog} to the $k$ uncorrelated boxes (of which $\mu k$ have a good Bell value). Then by Lemma \ref{XOR-lemma}, the XOR of the bits from these boxes gives a perfectly free random bit for $k$ sufficiently large.

We begin by estimating the error in the de Finetti bound in Lemma \ref{deFinetti-bound2}. Let us apply Lemma \ref{deFinetti-bound2}  with  $t_i = kt$, for all $2 \leq i \leq k$, and block sizes according to $n_i^{1-\log{(1+2\varepsilon)}} = 8 \ln{(2)} k^3 t^3 n_{i-1}$. Then we obtain that 
\begin{equation}
\Pr_{(a_1, \dots, a_k) \sim \nu([n_1] \times \dots \times [n_k])} \left( T \geq \frac{1}{t} \right) \leq \frac{1}{t}.
\end{equation}
with $T$ given by Eq. (\ref{defT}). Thus with probability larger than $1-1/t$ over the $(a_1, \ldots, a_k)$ we have $T \leq \frac{1}{t}$. By Markov inequality we have that for such good choices of $(a_1, \ldots, a_k)$,
\begin{equation}
\Pr_{\substack{\textbf{u}^1_{< a_1}, \ldots, \textbf{u}^k_{< a_k} \sim \nu_{a_1, \ldots, a_k} \\ \textbf{x}^1_{< a_1}, \ldots, \textbf{x}^k_{< a_k} \sim P}} ( T' \geq \eta ) \leq \frac{T}{\eta} \leq \frac{1}{t \eta}.
\end{equation}
with
\begin{equation} \label{Tprimeound}
T' :=  \mathbb{E}_{\textbf{u}^1_{a_1}, \ldots, \textbf{u}^k_{a_k}} \left \Vert  q(\textbf{x}^1_{a_1}, \ldots, \textbf{x}^k_{a_k} | \textbf{u}^1_{a_1}, \ldots, \textbf{u}^k_{a_k}) - q(\textbf{x}^1_{a_1} | \textbf{u}^1_{a_1}) \otimes \ldots \otimes q(\textbf{x}^k_{a_k} | \textbf{u}^k_{a_k}) \right \Vert_1,
\end{equation}
Choosing $\eta = 1/\sqrt{t}$ we find that with probability larger than $1 - 1/t - 1/\sqrt{t}$, the $k$ chosen boxes and the chosen inputs will be such that $T' \leq 1/ \sqrt{t}$.

For the ideal product state $q(\textbf{x}^1_{a_1} | \textbf{u}^1_{a_1}) \otimes \ldots \otimes q(\textbf{x}^k_{a_k} | \textbf{u}^k_{a_k})$, when the verification test in step $3$ of the protocol (Fig.~\ref{protocol}) accepts, we may infer that $\mu k$ of the boxes have good Bell value with probability $1 - \exp\left( -k \left( \frac{1}{2} - \varepsilon  \right)^8 (1 - \mu)^2 \delta^2/  8 \right)$. Then applying Lemma \ref{lin-prog} and Lemma \ref{XOR-lemma} we find 
\begin{equation}
d(\textsl{m}_{1} \oplus \dots \oplus \textsl{m}_{k} | \mathbb{Z}, \mathbb{W}) \leq \frac{1}{2} \left( \frac{11 + 7 \delta}{16} \right)^{\mu k} \left( 1 - e^{-\frac{k \left( \frac{1}{2} - \varepsilon  \right)^8 (1 - \mu)^2 \delta^2}{8} } \right) + e^{-\frac{k \left( \frac{1}{2} - \varepsilon  \right)^8 (1 - \mu)^2 \delta^2}{8} }.
\end{equation}

Therefore by Eq. (\ref{Tprimeound}), if the test accepts
\begin{equation}
d(\textsl{m}_{1} \oplus \dots \oplus \textsl{m}_{k} | \mathbb{Z}, \mathbb{W}) \leq \frac{1}{2} \left( \frac{11 + 7 \delta}{16} \right)^{\mu k}  + e^{-\frac{k \left( \frac{1}{2} - \varepsilon  \right)^8 (1 - \mu)^2 \delta^2}{8} } + 2/\sqrt{t}.
\end{equation}
Eq. (\ref{no-signaling-dist}) follows from Lemma \ref{conversioncomposable}.

The robustness of the protocol follows from Lemma \ref{robustnesssection}.
\end{proof}

%As we have seen, the protocol outputs a random bit uncorrelated with any no-signaling adversary, when the errors in the verification procedure and the de Finetti bounds are made small by choosing a large number of devices $k$ and large number of runs with each device $n_j$. Moreover, from Lemma \ref{prot-1-quantum}, we know that when the protocol is implemented with a quantum state 
%\begin{equation}
%\rho_1 = (1 - \tilde{\lambda}) |\Psi \rangle \langle \Psi| + \tilde{\lambda} \frac{\mathbb{I}}{16},
%\end{equation}
%with fraction $\tilde{\lambda}$ of maximally mixed state, with high probability the protocol accepts if 
%\begin{equation}
%\tilde{\lambda} \leq \frac{(1-\mu) \delta (\frac{1}{2} - \varepsilon)^4 \Vert P(\textbf{u})\Vert_8}{16 (\frac{1}{2} + \varepsilon)^4},
%\end{equation}
%where $\Vert P(\textbf{u}) \Vert_8 := \left(\frac{1}{2} + \varepsilon \right)^4 + \left( \frac{1}{2} - \varepsilon \right)^4 + 4 \left( \frac{1}{2} - %As can be seen from the above equation with constants $\delta = \frac{1}{2}$ and $\mu = \frac{1}{2}$, when $\varepsilon \rightarrow \frac{1}{2}$ we have $\tilde{\lambda} \rightarrow 0$ and the protocol needs to be implemented with the pure state $|\Psi \rangle$, while for any finite $\varepsilon < \frac{1}{2}$, it tolerates a fixed noise $\tilde{\lambda}$. 

\section{Tools for the correctness proofs}

\subsection{Randomness Amplification from Uncorrelated Good Devices}
\label{uncorrelated-boxes}
The first step in the proof of the protocol is the particular case where one has $m$ uncorrelated boxes $P_1(\textbf{x}_1 | \textbf{u}_1), \ldots, P_m(\textbf{x}_m | \textbf{u}_m)$ which are good in the sense that for all $1 \leq j \leq m$
\begin{equation} 
\label{goodindividual}
\textbf{B}. \{ P_j(\textbf{x}_j | \textbf{u}_j) \} \leq \delta,
\end{equation}
for some $\delta \geq 0$. We first show that the predictability of the majority bit obtained from a good box is bounded from above by a linear function of $\delta$. Then considering a large number $m$ of such good boxes that are uncorrelated from each other, we show that the XOR of the majority bits from these boxes can give rise to a perfect random bit. 

\begin{lemma} \label{lin-prog}
Consider a $4$-partite no-signaling box $P (\textbf{x}| \textbf{u})$ satisfying Eq. (\ref{goodindividual})
for some $\delta \geq 0$. Let $\textsl{m} := \text{maj}(x^1, x^2, x^3)$. Then
\begin{equation} \label{prod-dist}
d\left( \textsl{m} | \mathbb{Z}, \mathbb{W}  \right) \leq  \frac{1}{2}\left( \frac{11+ 7 \delta}{16}  \right).          
\end{equation}
\end{lemma}

\begin{proof}
In order to prove the lemma, we first formulate the distance from uniform of the majority bit $\textsl{m}$ as a linear program and obtain an upper bound on the distance from a feasible solution of the dual program. The distance from uniform in Eq. (\ref{distance-exp}) can be obtained by the following linear program
\begin{eqnarray}
\label{lin-prog1}
d \left(\textsl{m} | \mathbb{Z}, \mathbb{W} \right) &=& \max_{ \{ P \}}: \frac{1}{2} \textit{M}^T .\{ P(\textbf{x} | \textbf{u}) \} \nonumber \\
&&s.t. \; \; \textit{A}. \{ P( \textbf{x} | \textbf{u}) \} \leq \textit{c}.
\end{eqnarray}
Here, the indicator vector $\textit{M}$ is a $2^4 \times 2^4$ element vector with entries 
\begin{equation}
M(\textbf{x}, \textbf{u}) = \texttt{I}_{\textbf{u} = \textbf{u}^*} \texttt{I}_{\text{maj}(x^1, x^2, x^3) = 0} -\texttt{I}_{\textbf{u} = \textbf{u}^*} \texttt{I}_{\text{maj}(x^1, x^2, x^3) = 1}
\end{equation}
for any chosen measurement setting $\textbf{u}^* \in \textbf{U}_0 \cup \textbf{U}_1$. Analogous programs can be formulated for each of the $8$ measurement settings appearing in the Bell inequality in Eq. (\ref{Bell-ineq}). The constraint on the box $\{P(\textbf{x} | \textbf{u})\}$ written as a vector with $2^4 \times 2^4$ entries is given by the matrix $\textit{A}$ and the vector $\textit{c}$. These encode the no-signaling constraints between the four parties, the normalization and the positivity constraints on the probabilities $P(\textbf{x} | \textbf{u})$. In addition, $\textit{A}$ and $\textit{c}$ also encode the condition that $\textbf{B}.\{ P(\textbf{x}| \textbf{u}) \} \leq \delta$ where $\delta$ is a parameter that encodes the Bell value for the box. 

The solution to the primal linear program in Eq. (\ref{lin-prog1}) can be bounded by any feasible solution to the dual program which is written as
\begin{eqnarray}
\label{dual-lin-prog1}
&&\min_{ \lambda}: \textit{c}^T \lambda \nonumber \\
&& s.t. \; \; \; \textit{A}^T  \lambda = \textit{M}, \nonumber \\
&&\; \; \; \; \; \; \; \;  \lambda \geq 0.
\end{eqnarray}
Any vector $\lambda^*$ that is feasible in the sense that it satisfies the constraints to the dual program above gives an upper bound $\textit{c}^T \lambda^*$ to the distance from uniform, i.e. 
\begin{equation}
d \left(\textsl{m} | \mathbb{Z}, \mathbb{W} \right) \leq \frac{1}{2} \textit{c}^T \lambda^*,
\end{equation}
provided $\lambda^*$ satisfies the constraints in Eq. (\ref{dual-lin-prog1}). In the Appendix, we explicitly show such a feasible $\lambda^*$ for the measurement setting $\{u^1 u^2 u^3 u^4\} = \{0001\}$ that gives $\textit{c}^T \lambda^* = \left( \frac{11+ 7 \delta}{16}  \right)$. Similar feasible $\lambda^*$ can be found for all the $8$ measurement settings appearing in the Bell inequality which therefore gives 
\begin{equation}
d \left(\textsl{m} | \mathbb{Z}, \mathbb{W} \right) \leq \frac{1}{2} \left( \frac{11+ 7 \delta}{16}  \right).
\end{equation}
This completes the proof.
\end{proof}

Assume that we obtain $m$ independent boxes satisfying Eq. (\ref{goodindividual}). From Lemma \ref{lin-prog}, we know that the predictability of the majority bits obtained from such boxes is bounded. We now show that the XOR of such majority bits leads to a bit that is perfectly random in the limit of large $m$. 

\begin{lemma}
\label{XOR-lemma}
Let $X_i$, $i=1,\ldots,m$ be independent binary random variables with $P(X_i=0)= p_i$
satisfying $\frac12-\ep_i\leq p_i\leq \frac12 +\ep_i$. 
Then  $p=P(X_1 \oplus \ldots X_m=0)$ satisfies $\frac12 - \ep \leq p\leq \frac 12 +\ep $
with $\ep=2^{m-1} \prod_{i=1}^m \ep_i$.   
\end{lemma}
\begin{proof} Consider first the case $m=2$ and set  $X=X_1$, $Y=X_2$. We get that 
\ben
&& p\equiv P(X \oplus Y=0)=P(X=0,Y=0)+P(X=1,Y=1)=  \nonumber \\
&&P(X=0)P(Y=0)+P(X=1)P(Y=1) =p_1 p_2 + (1-p_1)(1-p_2)
\een
Note that global extrema of the above function of $p_1$ and $p_2$ are on the boundary. 
Indeed, for any fixed $p_1$, the function is linear in $p_2$, hence optimal $p_2$ are the extremal ones.
In turn, for any of the two extremal values of $p_2$, the function is again linear in $p_1$, so 
that optimal are extremal $p_1$. Thus we have to check four cases: $(p_1, p_2)$ being $(\frac12 \pm \ep_1, \frac12 \pm \ep_2)$.
This gives 
\be
\frac12 -\ep \leq p \leq \frac12 +\ep
\ee
where $\ep=2 \ep_1 \ep_2$.
Thus the formula is valid for $m=2$. For arbitrary $m$, the result follows from induction, by assuming 
that $X_1 \oplus \ldots \oplus X_{m-1}$ satisfies the formula for $m-1$, setting $X=X_1 \oplus \ldots \oplus X_{m-1}$, $Y=X_m$ and applying the result for two random variables. 
\end{proof}

\subsection{Imposing Independence: de Finetti bounds}
\label{deFinetti-section}
In this section, we show that a set of $k$ boxes chosen one from each device using an $\varepsilon$-SV source is close (in trace distance) to being uncorrelated for some suitable choice of block sizes $n_j$. The Lemmas in this section are inspired by the information-theoretic approach of \cite{Brandao} for proving de Finetti theorems for quantum states and non-signaling distributions.

% We first consider the case of the time-ordered non-signaling distributions $P(\textbf{x}_{\leq n_1}^1, \ldots, \textbf{x}_{\leq n_k}^k | \textbf{u}_{\leq n_1}^1, \ldots, \textbf{u}_{\leq n_k}^k)$ from Eq. \ref{to-ns}, the bounds derives in this scenario can be then translated to the special case when one imposes no-signaling conditions on all devices. 

\begin{lemma}
\label{deFinetti-bound2}
Let $P(\textbf{x}_{\leq n_1}^1, \ldots, \textbf{x}_{\leq n_k}^k | \textbf{u}_{\leq n_1}^1, \ldots, \textbf{u}_{\leq n_k}^k)$ be a time-ordered non-signaling distribution, with output and input alphabets $\Sigma$ and $\Lambda$, respectively (i.e. $P : \Sigma^{\times n} \times \Lambda^{\times n} \rightarrow \mathbb{R}^+$). The distribution $P$ represents $k$ devices each with $n_j$ devices. Let $(a_1, \ldots, a_k) \in [n_1] \times \ldots \times [n_k]$ and $(\textbf{u}_{\leq n_1}^1, \ldots, \textbf{u}_{\leq n_k}^k)$ be chosen from a $\varepsilon$-SV source $\nu(a_1, \ldots, a_k, \textbf{u}_{\leq n_1}^1, \ldots, \textbf{u}_{\leq n_k}^k)$. 

Then for every set of positive reals $\{ t_2, \ldots, t_k \}$,
%Assume $(a_1, \ldots, a_k) \in [n_1] \times \ldots \times [n_k]$ are chosen using an $\varepsilon$-SV source $\nu$ that is uncorrelated with the box $P$. Then for any probability distribution $\nu(\textbf{u}^1_1, \ldots, \textbf{u}^k_{n_k})$ and for every set of positive reals $\{t_2, \ldots, t_k\}$, we have
\begin{equation}
\label{deFinetti}
\Pr_{(a_1, \ldots, a_k) \sim \nu} \left( T \geq \sum_{i=2}^{k} \sqrt{\frac{8 \ln{(2)} t_i^2 \sum_{j=1}^{i-1} n_j}{n_i^{1-\log{(1+2\varepsilon)}}}} \right) \leq \sum_{i=2}^{k} \frac{1}{t_i},
\end{equation}
with
\begin{eqnarray} \label{defT}
T := \mathbb{E}_{\textbf{u}_{\leq n_1}^1, \ldots, \textbf{u}_{\leq n_k}^k \sim \nu_{a_1, \ldots, a_k}} &&\mathbb{E}_{\textbf{x}_{< a_1}^1, \ldots \textbf{x}_{< a_k}^k \sim P} \nonumber \\
&& \left \Vert q(\textbf{x}^1_{a_1}, \ldots, \textbf{x}^k_{a_k} | \textbf{u}^1_{a_1}, \ldots, \textbf{u}^k_{a_k}) - q(\textbf{x}^1_{a_1} | \textbf{u}^1_{a_1}) \otimes \ldots \otimes q(\textbf{x}^k_{a_k} | \textbf{u}^k_{a_k}) \right \Vert_1
\end{eqnarray}
where $q$ is the conditional box given the inputs $\textbf{u}_{< a_j}^j$ and outputs $\textbf{x}_{< a_j}^j$ for all $1 \leq j \leq k$
\begin{equation}
q(\textbf{x}^1_{a_1}, \ldots, \textbf{x}^k_{a_k} | \textbf{u}^1_{a_1}, \ldots, \textbf{u}^k_{a_k}) := P_
{\substack{\textbf{x}_{< a_1}^1, \dots, \textbf{x}_{< a_k}^k \\ \textbf{u}_{< a_1}^1, \dots, \textbf{u}_{< a_k}^k}} (\textbf{x}^1_{a_1}, \ldots, \textbf{x}^k_{a_k} | \textbf{u}^1_{a_1}, \ldots, \textbf{u}^k_{a_k}),
\end{equation}
and $\nu_{a_1, \ldots, a_k}$ is the probability $\nu$ conditioned on measuring $(a_1, \ldots, a_k)$.
\end{lemma}
%Moreover, there exist suitable choice of block sizes $n_j$ and reals $\{t_i\}$ such that
%\begin{equation}
%\label{simple-deFinetti2}
%\Pr_{(a_1, \ldots, a_k) \sim \nu} \left( T \geq \frac{1}{t'} \right) \leq \frac{1}{t'}
%\end{equation}
%for some suitably large parameter $t'$.
%\end{lemma}
%Again, the proof of the above lemma follows the same steps as that of Lemma \ref{deFinetti-bound} with the inputs $\textbf{u}_{< a_j}^j$ and outputs $\textbf{x}_{< a_j}^j$ now understood to have occurred earlier in time than $\textbf{x}_{a_j}^j$ for each $1 \leq j \leq k$. 

\begin{proof}
We first analyze the case when the $k$ boxes $(a_1, \ldots, a_k)$ are chosen from the uniform distribution over $[n_1] \times \ldots \times [n_k]$ and then consider the scenario where they are chosen from an $\varepsilon$-SV source. With $(a_1, \ldots, a_k)$ chosen uniformly, we first show that each of the $j$ boxes is approximately in a product state with the previous $j-1$ boxes. The product form of all boxes will then follow by application of the triangle inequality. 

Using the upper bound on mutual information $I(A:B) \leq \min (\log{|A|}, \log{|B|})$ and the chain rule $I(A: B C) = I(A:B) + I(A:C | B)$, we have
\begin{eqnarray}
\log |\Sigma| \sum_{j=1}^{k-1} n_j &\geq &  \mathbb{E}_{\textbf{u}^1_{\leq n_1}, \ldots, \textbf{u}^k_{\leq n_k} \sim \nu} I(\textbf{x}^1_{\leq n_1}, \ldots, \textbf{x}^{k-1}_{\leq n_{k-1}} : \textbf{x}^k_{\leq n_k})_P \nonumber \\
&=&\mathbb{E}_{\textbf{u}^1_{\leq n_1}, \ldots, \textbf{u}^k_{\leq n_k} \sim \nu} \left( I(\textbf{x}^1_{\leq n_1}, \ldots, \textbf{x}^{k-1}_{\leq n_{k-1}} : \textbf{x}^k_1)_P + \ldots + I(\textbf{x}^1_{\leq n_1}, \ldots, \textbf{x}^{k-1}_{\leq n_{k-1}} : \textbf{x}^k_{n_k} | \textbf{x}^k_{\leq n_k - 1})_P \right) \nonumber \\
\end{eqnarray}
Therefore, when $a_k$ is chosen from the uniform distribution $\mathit{U}[n_k]$, we get 
\begin{equation} \label{mutual-k}
\mathbb{E}_{a_k \sim \mathit{U}[n_k]} \mathbb{E}_{\textbf{u}^1_{\leq n_1}, \ldots, \textbf{u}^k_{\leq n_k} \sim \nu} \mathbb{E}_{\textbf{x}^{k}_{< a_k} \sim P} I(\textbf{x}^1_{\leq n_1}, \ldots, \textbf{x}^{k-1}_{\leq n_{k-1}} : \textbf{x}^{k}_{a_k})_q \leq \frac{\log |\Sigma|  \sum_{j=1}^{k-1} n_j}{n_k}. 
\end{equation}
We now use Pinsker's inequality relating the mutual information and trace distance as 
\begin{equation}
\mathbf{I}(A : B)_p \geq \frac{1}{2 \ln{(2)}} \left \Vert p_{A,B} - p_{A} \otimes p_{B} \right \Vert^2_1,
\end{equation}
and the convexity of $x^2$ to obtain 
\begin{eqnarray}
\label{prod-k}
&&\mathbb{E}_{a_k \sim \mathit{U}[n_k]} \mathbb{E}_{\textbf{u}^1_{\leq n_1}, \ldots, \textbf{u}^k_{\leq n_k} \sim \nu} \mathbb{E}_{\textbf{x}^{k}_{< a_k} \sim P} \nonumber \\
&&\left \Vert q(\textbf{x}^1_{\leq n_1}, \ldots, \textbf{x}^{k-1}_{\leq n_{k-1}}, \textbf{x}^k_{a_k} | \textbf{u}^1_{\leq n_1}, \ldots, \textbf{u}^{k-1}_{\leq n_{k-1}}, \textbf{u}^{k}_{a_k}) - q(\textbf{x}^1_{\leq n_1}, \ldots, \textbf{x}^{k-1}_{\leq n_{k-1}} | \textbf{u}^1_{\leq n_1}, \ldots, \textbf{u}^{k-1}_{\leq n_{k-1}}) \otimes q(\textbf{x}^{k}_{a_k} | \textbf{u}^k_{a_k}) \right \Vert_1 \nonumber \\
&& \hspace{6cm} \leq \sqrt{\frac{2 \ln{(2)} \log |\Sigma| \sum_{j=1}^{k-1} n_j}{n_k}}
\end{eqnarray}
The above argument can also be applied to the box $q(\textbf{x}^1_{\leq n_1}, \ldots, \textbf{x}^{k-2}_{\leq n_{k-2}}, \textbf{x}^{k-1}_{a_{k-1}} | \textbf{u}^1_{\leq n_1}, \ldots, \textbf{u}^{k-2}_{\leq n_{k-2}}, \textbf{u}^{k-1}_{a_{k-1}})$ to give
\begin{eqnarray}
\label{prod-k-1}
&&\mathbb{E}_{a_{k-1}, a_k \sim \mathit{U}([n_{k-1}] \times [n_k])} \mathbb{E}_{\textbf{u}^1_{\leq n_1}, \ldots, \textbf{u}^k_{\leq n_k} \sim \nu} \mathbb{E}_{\textbf{x}^{k-1}_{< a_{k-1}}, \textbf{x}^{k}_{< a_k} \sim P} \nonumber \\
&&\left \Vert q(\textbf{x}^1_{\leq n_1}, \ldots, \textbf{x}^{k-2}_{\leq n_{k-2}}, \textbf{x}^{k-1}_{a_{k-1}} | \textbf{u}^1_{\leq n_1}, \ldots, \textbf{u}^{k-2}_{\leq n_{k-2}}, \textbf{u}^{k-1}_{a_{k-1}}) - q(\textbf{x}^1_{\leq n_1}, \ldots, \textbf{x}^{k-2}_{\leq n_{k-2}} | \textbf{u}^1_{\leq n_1}, \ldots, \textbf{u}^{k-2}_{\leq n_{k-2}}) \otimes q(\textbf{x}^{k-1}_{a_{k-1}} | \textbf{u}^{k-1}_{a_{k-1}}) \right \Vert_1 \nonumber \\
&& \hspace{6cm} \leq \sqrt{\frac{2 \ln{(2) \log |\Sigma|} \sum_{j=1}^{k-2} n_j}{n_{k-1}}}
\end{eqnarray}

Using Eqs. (\ref{prod-k}) and (\ref{prod-k-1}), the triangle inequality and the monotonicity of the $1$-norm under discarding subsystems, we obtain 
\begin{eqnarray}
&&\mathbb{E}_{a_{k-1}, a_k \sim \mathit{U}([n_{k-1}] \times [n_k])} \mathbb{E}_{\textbf{u}^1_{\leq n_1}, \ldots, \textbf{u}^k_{\leq n_k} \sim \nu} \mathbb{E}_{\textbf{x}^{k-1}_{< a_{k-1}}, \textbf{x}^{k}_{< a_k} \sim P} \nonumber \\
&& \Vert q(\textbf{x}^1_{\leq n_1}, .., \textbf{x}^{k-2}_{\leq n_{k-2}}, \textbf{x}^{k-1}_{a_{k-1}}, \textbf{x}^{k}_{a_k} | \textbf{u}^1_{\leq n_1}, .., \textbf{u}^{k-2}_{\leq n_{k-2}}, \textbf{u}^{k-1}_{a_{k-1}}, \textbf{u}^{k}_{a_k}) - \nonumber \\ && \hspace{5cm} q(\textbf{x}^1_{\leq n_1}, .., \textbf{x}^{k-2}_{\leq n_{k-2}} | \textbf{u}^1_{\leq n_1}, .., \textbf{u}^{k-2}_{\leq n_{k-2}}) \otimes q(\textbf{x}^{k-1}_{a_{k-1}} | \textbf{u}^{k-1}_{a_{k-1}} ) \otimes q(\textbf{x}^{k}_{a_{k}} | \textbf{u}^{k}_{a_{k}})  \Vert_1 \nonumber \\
 &&\hspace{6cm} \leq \sqrt{\frac{2 \ln{(2) \log |\Sigma| } \sum_{j=1}^{k-1} n_j}{n_{k}}} + \sqrt{\frac{2 \ln{(2) \log |\Sigma|  } \sum_{j=1}^{k-2} n_j}{n_{k-1}}}
\end{eqnarray}

Following the above reasoning for the $(k-2)^{th}$ box up to the second we find
\begin{eqnarray}
\mathbb{E}_{a_{1}, \dots, a_k \sim \mathit{U}([n_1] \times \dots \times [n_k])} &&\mathbb{E}_{\textbf{u}^1_{\leq n_1}, \ldots, \textbf{u}^k_{\leq n_k} \sim \nu} \mathbb{E}_{\textbf{x}^{1}_{< a_{1}}, \dots, \textbf{x}^{k}_{< a_k} \sim P} \nonumber \\
&&\left \Vert q(\textbf{x}^1_{a_1}, \ldots, \textbf{x}^{k}_{a_k} | \textbf{u}^1_{a_1}, \ldots, \textbf{u}^{k}_{a_k}) - q(\textbf{x}^1_{a_1} | \textbf{u}^1_{a_1}) \otimes \ldots \otimes q(\textbf{x}^{k}_{a_{k}} | \textbf{u}^{k}_{a_{k}}) \right \Vert_1 \nonumber \\
&& \hspace{4cm} \leq \sum_{i=2}^{k} \sqrt{\frac{2 \ln{(2) \log |\Sigma| } \sum_{j=1}^{i-1} n_j}{n_{i}}}
\end{eqnarray}

%Using Eq. (\ref{mutual-k}) and the Markov inequality, we find
%\begin{equation}
%\Pr_{a_i \sim \mathit{U}[n_i]} \left( T_i \geq \sqrt{\frac{8 \ln{(2)} t_i \sum_{j=1}^{i-1} n_j}{n_i}} \right) \leq \frac{1}{t_i}
%\end{equation}
%with 
%\begin{eqnarray}
%T_i &:=&\mathbb{E}_{\textbf{u}^1_{\leq n_1}, \ldots, \textbf{u}^k_{\leq n_k} \sim \nu} \mathbb{E}_{\textbf{x}^{i}_{< a_{i}}, \ldots, %\textbf{x}^{k}_{< a_k} \sim P} \nonumber \\%
%&& \left \Vert q(\textbf{x}^1_{\leq n_1}, \ldots, \textbf{x}^{i-1}_{\leq n_{i-1}}, \textbf{x}^{i}_{a_i} | \textbf{u}^1_{\leq n_1}, \ldots, %\textbf{u}^{i-1}_{\leq n_{i-1}}, \textbf{u}^i_{a_i}) - q(\textbf{x}^1_{\leq n_1}, \ldots, \textbf{x}^{i-1}_{\leq n_{i-1}} | \textbf{u}^1_{\leq% n_1}, \ldots, \textbf{u}^{i-1}_{\leq n_{i-1}}) \otimes q(\textbf{x}^{i}_{a_{i}} | \textbf{u}^{i}_{a_{i}}) \right \Vert_1  \nonumber \\
%\end{eqnarray}

Let us now show how to extend the argument to the case when $(a_1, \ldots, a_k)$ are chosen from a $\varepsilon$-SV source. Let
\begin{equation}
N_i :=  I(\textbf{x}^1_{\leq n_1}, \ldots, \textbf{x}^{i-1}_{\leq n_{i-1}} : \textbf{x}^i_{a_i})_{q} .
\end{equation}
From the chain-rule argument presented before we have
\begin{eqnarray}
&& \mathbb{E}_{a_i \sim U[n_i]}  \mathbb{E}_{(a_1, \ldots, a_{i-1}, a_{i+1}, \dots, a_k) \sim \nu_{a_i}}   \mathbb{E}_{\textbf{u}^1_{\leq n_1}, \ldots, \textbf{u}^k_{\leq n_k} \sim \nu_{a_1, \ldots, a_k}} \mathbb{E}_{\textbf{x}^{k}_{< a_i, \ldots, < a_k} \sim P} \hspace{0.1 cm} N_i\nonumber \\ &\leq& \frac{ \log |\Sigma| \sum_{j=1}^{i-1} n_j}{n_i}.
\end{eqnarray}
Then from the definition of a $\varepsilon$-SV source,
\begin{eqnarray}
\label{prod-k}
&& \mathbb{E}_{(a_1, \ldots a_k) \sim \nu} \mathbb{E}_{\textbf{u}^1_{\leq n_1}, \ldots, \textbf{u}^k_{\leq n_k} \sim \nu_{a_1, \ldots, a_k}} \mathbb{E}_{\textbf{x}^{k}_{< a_i, \ldots, < a_k} \sim P}  \hspace{0.1 cm}  N_i \nonumber \\ 
&=& \mathbb{E}_{a_j \sim \nu}  \mathbb{E}_{(a_1, \ldots, a_{j-1}, a_{j+1}, \ldots, a_k) \sim \nu_{a_j}}   \mathbb{E}_{\textbf{u}^1_{\leq n_1}, \ldots, \textbf{u}^k_{\leq n_k} \sim \nu_{a_1, \ldots, a_k}} \mathbb{E}_{\textbf{x}^{k}_{< a_i, \ldots, < a_k} \sim P}  \hspace{0.1 cm}  N_i  \nonumber \\ 
&\leq& n_i \left( \frac{1}{2} + \varepsilon \right)^{\log(n_i)}  \mathbb{E}_{a_j \sim U[n_i]}  \mathbb{E}_{(a_1, \ldots, a_{j-1}, a_{j+1}, \dots, a_k) \sim \nu_{a_j}}   \mathbb{E}_{\textbf{u}^1_{\leq n_1}, \ldots, \textbf{u}^k_{\leq n_k} \sim \nu_{a_1, \ldots, a_k}} \mathbb{E}_{\textbf{x}^{k}_{< a_i, \ldots, < a_k} \sim P}  \hspace{0.1 cm} N_i  \nonumber \\ 
&\leq&  n_i \left( \frac{1}{2} + \varepsilon \right)^{\log(n_i)}   \frac{ \log |\Sigma| \sum_{j=1}^{i-1} n_j}{n_i}.
%\nonumber \\
%&&\left \Vert q(\textbf{x}^1_{\leq n_1}, \ldots, \textbf{x}^{i-1}_{\leq n_{i-1}}, \textbf{x}^i_{a_i} | \textbf{u}^1_{\leq n_1}, \ldots, \textbf{u}^{i-1}_{\leq n_{i-1}}, \textbf{u}^{i}_{a_i}) - q(\textbf{x}^1_{\leq n_1}, \ldots, \textbf{x}^{i-1}_{\leq n_{i-1}} | \textbf{u}^1_{\leq n_1}, \ldots, \textbf{u}^{i-1}_{\leq n_{i-1}}) \otimes q(\textbf{x}^{i}_{a_i} | \textbf{u}^i_{a_i}) \right \Vert_1 \nonumber \\
%&& \hspace{6cm} \leq \sqrt{\frac{2 \ln{(2)} \log |\Sigma| \sum_{j=1}^{i-1} n_j}{n_i}}
\end{eqnarray}
Then by Pinsker's inequality and the convexity of $x^2$,
\begin{eqnarray}
&&\mathbb{E}_{(a_1, \ldots a_k) \sim \nu} \mathbb{E}_{\textbf{u}^1_{\leq n_1}, \ldots, \textbf{u}^k_{\leq n_k} \sim \nu_{a_1, \ldots, a_k}} \mathbb{E}_{\textbf{x}^{k}_{< a_i, \ldots, < a_k} \sim P}  \nonumber \\
&& \left \Vert q(\textbf{x}^1_{\leq n_1}, \ldots, \textbf{x}^{i-1}_{\leq n_{i-1}}, \textbf{x}^i_{a_i} | \textbf{u}^1_{\leq n_1}, \ldots, \textbf{u}^{i-1}_{\leq n_{i-1}}, \textbf{u}^{i}_{a_i}) - q(\textbf{x}^1_{\leq n_1}, \ldots, \textbf{x}^{i-1}_{\leq n_{i-1}} | \textbf{u}^1_{\leq n_1}, \ldots, \textbf{u}^{i-1}_{\leq n_{i-1}}) \otimes q(\textbf{x}^{i}_{a_i} | \textbf{u}^i_{a_i}) \right \Vert_1  \nonumber \\
&\leq&  \sqrt{2 \ln(2) n_i^{\log(1+2\varepsilon)} \frac{ \log |\Sigma| \sum_{j=1}^{i-1} n_j}{n_i}}.
%\nonumber \\
%&&\left \Vert q(\textbf{x}^1_{\leq n_1}, \ldots, \textbf{x}^{i-1}_{\leq n_{i-1}}, \textbf{x}^i_{a_i} | \textbf{u}^1_{\leq n_1}, \ldots, \textbf{u}^{i-1}_{\leq n_{i-1}}, \textbf{u}^{i}_{a_i}) - q(\textbf{x}^1_{\leq n_1}, \ldots, \textbf{x}^{i-1}_{\leq n_{i-1}} | \textbf{u}^1_{\leq n_1}, \ldots, \textbf{u}^{i-1}_{\leq n_{i-1}}) \otimes q(\textbf{x}^{i}_{a_i} | \textbf{u}^i_{a_i}) \right \Vert_1 \nonumber \\
%&& \hspace{6cm} \leq \sqrt{\frac{2 \ln{(2)} \log |\Sigma| \sum_{j=1}^{i-1} n_j}{n_i}}
\end{eqnarray}
%with $M_i$ defined as 
%\begin{eqnarray}
%\left \Vert q(\textbf{x}^1_{\leq n_1}, \ldots, \textbf{x}^{i-1}_{\leq n_{i-1}}, \textbf{x}^i_{a_i} | \textbf{u}^1_{\leq n_1}, \ldots, \textbf{u}^{i-1}_{\leq n_{i-1}}, \textbf{u}^{i}_{a_i}) - q(\textbf{x}^1_{\leq n_1}, \ldots, \textbf{x}^{i-1}_{\leq n_{i-1}} | \textbf{u}^1_{\leq n_1}, \ldots, \textbf{u}^{i-1}_{\leq n_{i-1}}) \otimes q(\textbf{x}^{i}_{a_i} | \textbf{u}^i_{a_i}) \right \Vert_1. \nonumber 
%\end{eqnarray}

By Markov inequality,
\begin{equation}
\label{single-ineq}
\Pr_{(a_1, \ldots, a_k) \sim \nu} \left( T_i \geq \sqrt{\frac{2 \ln{(2) \log |\Sigma|} t_i \sum_{j=1}^{i-1} n_j}{n_i}} \right) \leq \sqrt{\frac{n_i^{\log{(1 + 2\varepsilon)}}}{t_i}},
\end{equation}
with
\begin{eqnarray}
&&T_i :=\mathbb{E}_{\textbf{u}^1_{\leq n_1}, \ldots, \textbf{u}^k_{\leq n_k} \sim \nu_{a_1, \ldots, a_{k}}} \mathbb{E}_{\textbf{x}^{i}_{< a_{i}}, \ldots, \textbf{x}^{k}_{< a_k} \sim P}  \\
&& \left \Vert q(\textbf{x}^1_{\leq n_1}, \ldots, \textbf{x}^{i-1}_{\leq n_{i-1}}, \textbf{x}^{i}_{a_i} | \textbf{u}^1_{\leq n_1}, \ldots, \textbf{u}^{i-1}_{\leq n_{i-1}}, \textbf{u}^i_{a_i}) - q(\textbf{x}^1_{\leq n_1}, \ldots, \textbf{x}^{i-1}_{\leq n_{i-1}} | \textbf{u}^1_{\leq n_1}, \ldots, \textbf{u}^{i-1}_{\leq n_{i-1}}) \otimes q(\textbf{x}^{i}_{a_{i}} | \textbf{u}^{i}_{a_{i}}) \right \Vert_1  \nonumber 
\end{eqnarray}

However, by the triangle inequality and the monotonicity of the $1$-norm under discarding subsystems, we have that if for all $2 \leq i \leq k$, 
\begin{equation}
T_i < \sqrt{\frac{2 \ln{(2) \log|\Sigma|} t_i \sum_{j=1}^{i-1} n_j}{n_i}}  =: r_i ,
\end{equation}
then
\begin{equation}
T < \sum_{i=2}^{k} r_i.
\end{equation}
Therefore,
\begin{eqnarray}
\Pr_{(a_1, \dots, a_k) \sim \nu([n_1] \times \dots \times [n_k])} \left( T < \sum_{i=2}^{k} r_i \right) &\geq & \Pr_{(a_1, \dots, a_k) \sim \eta([n_1] \times \dots \times [n_k])} \left( T_2 < r_2 \cap \dots \cap T_k < r_k \right) \nonumber \\
& = & 1 - \Pr_{(a_1, \dots, a_k) \sim \nu([n_1] \times \dots \times [n_k])} \left( T_2 \geq r_2 \cup \dots \cup T_k \geq r_k \right) \nonumber \\
& \geq & 1 - \sum_{i=2}^{k} \sqrt{\frac{n_i^{\log{(1 + 2\varepsilon)}}}{t_i}},
\end{eqnarray}
where the final inequality follows from the union bound and Eq.(\ref{single-ineq}). This gives
\begin{eqnarray}
\label{k-ineq}
\Pr_{(a_1, \dots, a_k) \sim \nu([n_1] \times \dots \times [n_k])} \left( T \geq \sum_{i=2}^{k} \sqrt{\frac{8 \ln{(2)} t_i \sum_{j=1}^{i-1} n_j}{n_i}} \right) \leq \sum_{i=2}^{k} \sqrt{\frac{n_i^{\log{(1 + 2\varepsilon)}}}{t_i}}.
\end{eqnarray}
Finally, to obtain Eq.(\ref{deFinetti}), we replace $t_i$ by $t_i^2 n_i^{\log{(1+2\varepsilon)}}$ in the equation above 
\end{proof}

\subsection{Protocol verification procedure}
\label{verification}

In this section, we provide the proof that the verification procedure in the protocol works correctly, i.e. when the test accepts, a fraction $\mu$ of the tested boxes have good Bell value. The procedure consists in performing an estimation on the boxes denoted by $(a_1, \dots, a_k)$, chosen with the $\varepsilon$-SV source, and show that when the outcomes from these boxes pass a test, a fraction $m = \mu k$ of them have good Bell value. 

For simplicity, in this section we will use the notation $(1, \dots, k)$ in place of $(a_1, \dots, a_k)$. We say a no-signaling box $P(\textbf{x}_1, \ldots, \textbf{x}_k | \textbf{u}_1, \ldots, \textbf{u}_k)$ and a choice of inputs $(\textbf{u}_1, \ldots, \textbf{u}_k)$ and outputs  $(\textbf{x}_1, \ldots, \textbf{x}_k)$ are $(\mu, \delta)$-good if for all $l \in A$, with $A$ a subset of $\{ 1, \ldots, k \}$ of size larger than $\mu k$,
\begin{equation}
\textbf{B}.\{P_{\substack{\textbf{x}_1, \ldots, \textbf{x}_{l-1}\\ \textbf{u}_1, \ldots, \textbf{u}_{l-1}}}(\textbf{x}_l | \textbf{u}_l)\} < \delta. 
\end{equation}
Here $P_{\substack{\textbf{x}_1, \ldots, \textbf{x}_{l-1} \\ \textbf{u}_1, \ldots, \textbf{u}_{l-1}}}(\textbf{x}_l | \textbf{u}_l)$ is the box conditioned on the inputs $\textbf{u}_1, \ldots, \textbf{u}_{l-1}$ and the outcomes $\textbf{x}_1, \ldots, \textbf{x}_{l-1}$.

Let us consider the test where one computes the empirical violation average of the constraints in $\textbf{B}$:
\begin{equation}
Z_k := \frac{1}{k} \sum_{l=1}^{k} B(\textbf{x}_l, \textbf{u}_l),
\end{equation}
accepts if 
\begin{equation}
Z_k \leq \left( \frac{1}{2} - \varepsilon  \right)^4 (1 - \mu) \frac{\delta}{2},
\end{equation}
and rejects otherwise.

\begin{lemma}
\label{Azuma-estimation}
The test described above rejects with probability larger than $1 - \exp\left( -k \left( \frac{1}{2} - \varepsilon  \right)^8 (1 - \mu)^2 \delta^2/  8 \right)$ unless $((\textbf{u}_1, \ldots, \textbf{u}_k),  (\textbf{x}_1, \ldots, \textbf{x}_k), P(\textbf{x}_1, \ldots, \textbf{x}_k | \textbf{u}_1, \ldots, \textbf{u}_k))$ are $(\mu, \delta)$-good.
\end{lemma}

\begin{proof}
Assume that $((\textbf{u}_1, \ldots, \textbf{u}_k),  (\textbf{x}_1, \ldots, \textbf{x}_k), P(\textbf{x}_1, \ldots, \textbf{x}_k | \textbf{u}_1, \ldots, \textbf{u}_k))$  are not $(\mu, \delta)$-good. Let us show the test rejects with probability greater or equal to $1 - \exp\left( -k \left( \frac{1}{2} - \varepsilon  \right)^8 (1 - \mu)^2 \delta^2/  8 \right)$. 

Define
%\begin{equation}
%\delta_k := \mathbb{E}_{u_k \sim \text{uniform}} \mathbb{E}_{x_k \sim p_{u_1, \ldots, u_{k-1}}(x_k | %u_k)}B(x_k, u_k) .
%\end{equation}
\begin{equation}
\delta_l := \textbf{B}.\{P_{\substack{\textbf{x}_1, \ldots, \textbf{x}_{l-1} \\ \textbf{u}_1, \ldots, \textbf{u}_{l-1}}}(\textbf{x}_l | \textbf{u}_l) \}.
\end{equation}
Then by the definition of a $\varepsilon$-SV source we have that 
%for every function $f(\textbf{u}_1, \ldots, \textbf{u}_{k-1})$,
\begin{equation} \label{SVeffect}
\mathbb{E}_{\textbf{u}_l \sim \nu} \mathbb{E}_{\textbf{x}_l \sim P_{\substack{\textbf{x}_1, \ldots, \textbf{x}_{l-1}\\\textbf{u}_1, \ldots, \textbf{u}_{l-1}}}(\textbf{x}_l | \textbf{u}_l)} B(\textbf{x}_l, \textbf{u}_l) \geq \left( \frac{1}{2}  - \varepsilon \right)^4 \delta_l =: \zeta_l. 
\end{equation}

Let us define
\begin{equation}
X_l := \sum_{i=1}^l (\zeta_i - B(\textbf{x}_i, \textbf{u}_i)). 
\end{equation}

We claim $\{ X_1, \ldots, X_k \}$ form a supermartingale with respect to $\{(\textbf{x}_1, \textbf{u}_1), \ldots, (\textbf{x}_k, \textbf{u}_k)\}$. Indeed we have
\begin{equation}
\mathbb{E}(X_{l} |  (\textbf{x}_{1}, \textbf{u}_1), \ldots, (\textbf{x}_{l-1} , \textbf{u}_{l-1})) = X_{l-1} + \zeta_{l} - \mathbb{E}(B(\textbf{x}_l, \textbf{u}_l) |  (\textbf{x}_{1}, \textbf{u}_1), \ldots, (\textbf{x}_{l-1} , \textbf{u}_{l-1})) \leq X_{l-1}
\end{equation}
where the last inequality follows from Eq. (\ref{SVeffect}) as follows:
\begin{eqnarray}
\label{supermartingale}
&&\mathbb{E}(B(\textbf{x}_{l}, \textbf{u}_{l}) | (\textbf{x}_{1}, \textbf{u}_1), \ldots, (\textbf{x}_{l-1} , \textbf{u}_{l-1})) \nonumber \\ &&\quad = \sum_{\textbf{u}_l, \ldots, \textbf{u}_k} \sum_{\textbf{x}_{l} , \ldots, \textbf{x}_{k}} \nu(\textbf{u}_{l}, \ldots, \textbf{u}_{k} | (\textbf{x}_1, \textbf{u}_1), \ldots, (\textbf{x}_{l-1}, \textbf{u}_{l-1})) P_{\substack{\textbf{x}_1, \ldots, \textbf{x}_{l-1}\\\textbf{u}_1, \ldots, \textbf{u}_{l-1}}} (\textbf{x}_{l}, \ldots, \textbf{x}_{k} | \textbf{u}_1, \ldots, \textbf{u}_k) B(\textbf{x}_{l}, \textbf{u}_l) \nonumber \\ &&\quad = \sum_{\textbf{u}_{l}} \sum_{\textbf{x}_{l}} \nu(\textbf{u}_{l} | (\textbf{x}_1, \textbf{u}_1), \ldots, (\textbf{x}_{l-1}, \textbf{u}_{l-1})) P_{\substack{\textbf{x}_1, \ldots, \textbf{x}_{l-1}\\ \textbf{u}_1, \ldots, \textbf{u}_{l-1}}}(\textbf{x}_l | \textbf{u}_l) B(\textbf{x}_l, \textbf{u}_l) \nonumber \\ && \quad \geq \left( \frac{1}{2} - \varepsilon \right)^4 \sum_{\textbf{u}_{l}} \sum_{\textbf{x}_{l}} P_{\substack{\textbf{x}_{1}, \ldots, \textbf{x}_{l-1}\\ \textbf{u}_1, \ldots, \textbf{u}_{l-1}}} (\textbf{x}_l | \textbf{u}_l) B(\textbf{x}_l, \textbf{u}_l) = \left( \frac{1}{2} - \varepsilon \right)^4 \delta_l = \zeta_l.
\end{eqnarray}

Moreover $|X_l - X_{l-1}| \leq 1$ since $B(\textbf{x}_i, \textbf{u}_i) \in \{0, 1\}$ and $0 \leq \delta_l \leq 8$ (the maximum value of the Bell expression by normalization) giving $0\leq \zeta_l \leq \frac{1}{2}$. Thus by Azuma-Hoeffding inequality (Lemma \ref{azuma}) taking $X_0 = 0$ we find
\begin{equation}
\Pr \left(  X_k \geq t  \right) \leq e^{-t^2 / 2k}.
\end{equation}
Rearranging terms and defining $s := t/k$,
\begin{equation}
\Pr \left(  Z_k \leq \frac{1}{k}  \sum_{i=1}^{k} \zeta_i - s   \right) \leq e^{- s^2 k / 2}.
\end{equation}

Since $((\textbf{u}_1, \ldots, \textbf{u}_k),  (\textbf{x}_1, \ldots, \textbf{x}_k), P(\textbf{x}_1, \ldots, \textbf{x}_k | \textbf{u}_1, \ldots, \textbf{u}_k))$ are not $(\mu, \delta)$-good,
\begin{equation}
\sum_{i=1}^{k} \zeta_i  = \left( \frac{1}{2}  - \varepsilon \right)^4 \sum_{i=1}^{k} \delta_i \geq  \left( \frac{1}{2}  - \varepsilon \right)^4 (1 - \mu) k \delta,
\end{equation}
and so 
\begin{equation}
\Pr \left(  Z_k \leq \left( \frac{1}{2}  - \varepsilon \right)^4 (1 - \mu) \delta - s \right) \leq e^{- s^2 k / 2}.
\end{equation}
Taking $s = \left( \frac{1}{2} - \varepsilon \right)^4 (1 - \mu) \frac{\delta}{2}$ gives the result. 
\end{proof}

In the proof above, we used the notion of supermartingales and the associated Azuma-Hoeffding inequality which we recount here for convenience. Let $X_0, \ldots, X_k$ and $Y_0, \ldots, Y_k$ be two sequences of random variables. Then $X_0, \ldots, X_k$ is said to be a supermartingale with respect to $Y_0, \ldots, Y_k$ if for all $0 \leq i \leq k$, $\mathbb{E}|X_{i}| < \infty$ and $\mathbb{E}(X_{i} | Y_0, \ldots, Y_{i-1}) \leq X_{i-1}$. 

\begin{lemma} \label{azuma}
(Azuma-Hoeffding) Suppose $X_0, \ldots, X_k$ is a supermartingale with respect to $Y_0, \ldots, Y_k$, and that $|X_{l+1} - X_l| \leq c_l$ for all $0 \leq l \leq k-1$. Then for all positive reals $t$, 
\begin{equation}
\Pr \left( X_k - X_0 \geq t  \right) \leq \exp \left( - \frac{t^2}{2 \sum_{l=1}^{k} c_l^2}  \right).\\
\end{equation}
\end{lemma}

\subsection{Robustness of the Protocol} 

In the remainder of this section, we would like to estimate the amount of noise that the protocol can tolerate. Suppose we are given a box such that for every inputs and outputs, $((\textbf{u}_1, \ldots, \textbf{u}_k), (\textbf{x}_1, \ldots, \textbf{x}_k), P(\textbf{x}_1, \ldots, \textbf{x}_k | \textbf{u}_1, \ldots, \textbf{u}_k))$ are $(1, \tilde{\delta})$-good. This will be the case, for example, if all the the entangled states and measurements used to produce a box are only $O(\tilde{\delta})$-close to the ones that would lead to a box violating maximally the Bell inequality (i.e. $k$ copies of the entangled state given by Eq. (\ref{state}), each measured in the bases given by Eqs. (\ref{mea1}) and (\ref{mea2}).

%of white noise that can be added to the ideal quantum state $|\Psi \rangle$ used in the violation of the four-party inequality while still not aborting the protocol. Assume the protocol is performed using the quantum state 
%\begin{equation}
%\rho_1 = (1 - \tilde{\lambda}) |\Psi \rangle \langle \Psi| + \tilde{\lambda} \frac{\mathbb{I}}{16},
%\end{equation}
%with fraction $\tilde{\lambda}$ of maximally mixed state. The Bell value for the quantum box obtained from this state is given by 
%\begin{equation}
%\tilde{\delta} = \textbf{B}.\{P(\textbf{x} | \textbf{u})\}_{\rho_1} = \frac{\tilde{\lambda}}{4}.
%\end{equation}
%We now find the amount of white noise $\tilde{\lambda}$ as a function of the initial $\varepsilon$; as $\varepsilon \rightarrow \frac{1}{2}$ we expect to find $\tilde{\lambda} \rightarrow 0$. To this end, let us suppose that the verification procedure in section \ref{prot-1-verification} is implemented with all $k$ boxes having Bell value $\tilde{\delta}$.  

\begin{lemma} \label{robustnesssection}
\label{prot-1-quantum}
Consider the verification procedure applied a triple $((\textbf{u}_1, \ldots, \textbf{u}_k), (\textbf{x}_1, \ldots, \textbf{x}_k), P(\textbf{x}_1, \ldots, \textbf{x}_k | \textbf{u}_1, \ldots, \textbf{u}_k))$ which is $(1, \tilde{\delta})$-good. Then the test accepts with probability $1 - \exp{\left(-\frac{k (1/2 - \varepsilon)^8 (1 - \mu)^2 \delta^2}{2048}\right)}$ as long as 
\begin{equation}
\tilde{\delta} \leq \frac{(1-\mu) \delta (\frac{1}{2} - \varepsilon)^4 f(\varepsilon)}{4 (\frac{1}{2} + \varepsilon)^4},
\end{equation}
where $f(\varepsilon) := \left(\frac{1}{2} + \varepsilon \right)^4 + \left( \frac{1}{2} - \varepsilon \right)^4 + 4 \left( \frac{1}{2} - \varepsilon \right)^3 \left(\frac{1}{2} + \varepsilon \right) + 2 \left( \frac{1}{2} - \varepsilon \right)^2 \left( \frac{1}{2} + \varepsilon \right)^2$.
\end{lemma}

\begin{proof}
The proof follows similarly to that of Lemma \ref{Azuma-estimation}. 
We have that for all $1 \leq l \leq k$,
\begin{equation}
\delta_l := \textbf{B}.\{P_{\substack{\textbf{x}_1, \ldots, \textbf{x}_{l-1} \\ \textbf{u}_1, \ldots, \textbf{u}_{l-1}}}(\textbf{x}_l | \textbf{u}_l) \} \leq \tilde{\delta}.
\end{equation}
By the definition of the $\varepsilon$-SV source
\begin{equation} \label{SVeffect}
\mathbb{E}_{\textbf{u}_l \sim \nu} \mathbb{E}_{\textbf{x}_l \sim P_{\substack{\textbf{x}_1, \ldots, \textbf{x}_{l-1}\\\textbf{u}_1, \ldots, \textbf{u}_{l-1}}}(\textbf{x}_l | \textbf{u}_l)} B(\textbf{x}_l, \textbf{u}_l) \leq \frac{\left( \frac{1}{2}  + \varepsilon \right)^4}{f(\varepsilon)} \delta_l =: \tilde{\zeta}_l, 
\end{equation}
with the norm defined by $f(\varepsilon) := \left(\frac{1}{2} + \varepsilon \right)^4 + \left( \frac{1}{2} - \varepsilon \right)^4 + 4 \left( \frac{1}{2} - \varepsilon \right)^3 \left(\frac{1}{2} + \varepsilon \right) + 2 \left( \frac{1}{2} - \varepsilon \right)^2 \left( \frac{1}{2} + \varepsilon \right)^2$. This follows since $\frac{\left( \frac{1}{2}  + \varepsilon \right)^4}{f(\varepsilon)}$ is the maximum probability of any set of four measurements.

Defining 
\begin{equation}
\tilde{X}_l := \sum_{i=1}^l (B(\textbf{x}_i, \textbf{u}_i) - \tilde{\zeta}_i), 
\end{equation}
and following Eq. (\ref{supermartingale}), we find
\begin{equation}
\mathbb{E}(\tilde{X}_{l} |  (\textbf{x}_{1}, \textbf{u}_1), \ldots, (\textbf{x}_{l-1} , \textbf{u}_{l-1})) = \tilde{X}_{l-1} + \mathbb{E}(B(\textbf{x}_l, \textbf{u}_l) |  (\textbf{x}_{1}, \textbf{u}_1), \ldots, (\textbf{x}_{l-1} , \textbf{u}_{l-1})) - \tilde{\zeta}_{l} \leq \tilde{X}_{l-1}.
\end{equation}
In other words, $\{ \tilde{X}_1, \ldots, \tilde{X}_k \}$ form a supermartingale with respect to $\{(\textbf{x}_1, \textbf{u}_1), \ldots, (\textbf{x}_k, \textbf{u}_k)\}$.

Since $|\tilde{X}_l - \tilde{X}_{l-1}| \leq 8$ by Lemma \ref{azuma} (with $\tilde{X}_0 = 0$),
\begin{equation}
\Pr \left( Z_k \geq \frac{1}{k}  \sum_{i=1}^{k} \tilde{\zeta}_i + s   \right) \leq e^{- s^2 k /128}.
\end{equation}
When $((\textbf{u}_1, \ldots, \textbf{u}_k), (\textbf{x}_1, \ldots, \textbf{x}_k), P(\textbf{x}_1, \ldots, \textbf{x}_k | \textbf{u}_1, \ldots, \textbf{u}_k))$ are $(1, \tilde{\delta})$-good, we know that 
\begin{equation}
\sum_{i=1}^{k} \tilde{\zeta}_i = \frac{(\frac{1}{2} + \varepsilon)^4}{ f(\varepsilon) } \sum_{i=1}^{k} \delta_i =  \frac{(\frac{1}{2} + \varepsilon)^4}{f(\varepsilon)} k \tilde{\delta}
\end{equation}
Choosing 
%\begin{equation}
$s =\frac{ (\frac{1}{2} - \varepsilon)^4 (1 - \mu) \delta}{4},$
% \end{equation}
we find that when 
\begin{equation}
\tilde{\delta} \leq \frac{(1-\mu) \delta (\frac{1}{2} - \varepsilon)^4 f(\varepsilon) }{4 (\frac{1}{2} + \varepsilon)^4},
\end{equation}
we have
\begin{equation}
\Pr \left( Z_k \leq \frac{(1- \mu) \delta (\frac{1}{2} - \varepsilon)^4}{2}  \right) \geq 1 - e^{\left(- \frac{k (\frac{1}{2} - \varepsilon)^8 (1- \mu)^2 \delta^2} {2^{11}}\right)},
\end{equation}
%In other words, 
%\begin{equation}
%\Pr \left( Z_k \leq \frac{(1- \mu) \delta (\frac{1}{2} - \varepsilon)^4}{2}  \right) \geq 1 - e^{\left(- \frac{k %(\frac{1}{2} - \varepsilon)^8 (1- \mu)^2 \delta^2} {2^{11}}\right)},
%\end{equation}
so that the test is passed with high probability.
\end{proof}

\section{Conclusion and Open Questions} 

We have presented a protocol for obtaining secure random bits from an arbitrarily (but not fully) deterministic $\varepsilon$-SV source. The protocol uses correlations violating a four-party Bell inequality, includes an explicit hash function, and works even with correlations attainable by noisy quantum mechanical resources, producing a bit $\varepsilon'$-close to random with $\textit{O}(\log{(1/\varepsilon')})$ devices. Moreover the correctness of the protocol is not based on quantum mechanics and only requires the no-signalling principle.

We leave the following open questions to future research:

\begin{itemize}

\item Is there a randomness amplification protocol secure against no-signaling adversaries using only a finite number of devices? While such protocols may be formulated assuming a set of independent boxes (i.e. an ``individual attack'' by the eavesdropper), a proof for general coherent attacks is lacking. 

\item Can randomness amplification be based on a bipartite Bell inequality, i.e., are there bipartite Bell inequalities which allow algebraic violation by quantum correlations in addition to incorporating randomness?

\item Is there a protocol that can tolerate a higher level of noise? What if we assume the validity of quantum mechanics?

\item Can we amplify randomness from other different types of sources? A particularly interesting case is the min-entropy source \cite{Scarani}. 

\item A more technical question is to improve the de Finetti theorem given in \cite{Brandao, Brandao2}. What are the limits of de Finetti type results when the subsystems are selected from a Santha-Vazirani source?

\item Finally suppose one would like to realize device-independent quantum key distribution with only an imperfect SV source as the randomness source. Is there an efficient protocol for that tolerating a constant rate of noise and giving a constant rate of key? Here the question is open for both quantum-mechanical and non-signalling adversaries. 

\end{itemize}

{\it Acknowledgments.}
The paper is supported by ERC AdG grant QOLAPS and by Foundation for Polish Science TEAM project co-financed by the EU European Regional Development Fund. FB acknowledges support from EPSRC. Part of this work was done in National Quantum Information Center of Gda\'{n}sk.

\bibliographystyle{apsrev}

%\bibliography{rmp11-hugekey-phd}

\end{document}